\documentclass[12pt]{article}
\usepackage{geometry}
\usepackage{amsthm,amsmath,amssymb}
\usepackage{graphicx}
\usepackage{pgf,tikz}
\usepackage[authoryear]{natbib}
\usepackage{placeins}

\newtheorem{theorem}{Theorem}[section]
\newtheorem{lemma}{Lemma}[section]

\newtheorem{proposition}{Proposition}[section]

\theoremstyle{definition}

\begin{document}
\title{Odds are the sign is right}
\author{Brian Knaeble and Julian Chan} 
\date{2018}
\maketitle
\begin{abstract}
This article introduces a new condition based on odds ratios for sensitivity analysis.  The analysis involves the average effect of a treatment or exposure on a response or outcome with estimates adjusted for and conditional on a single, unmeasured, dichotomous covariate.  Results of statistical simulations are displayed to show that the odds ratio condition is as reliable as other commonly used conditions for sensitivity analysis.  Other conditions utilize quantities reflective of a mediating covariate.  The odds ratio condition can be applied when the covariate is a confounding variable.  As an example application we use the odds ratio condition to analyze and interpret a positive association observed between Zika virus infection and birth defects. 
\end{abstract}

\section{Introduction}
\label{int}
Simpson's paradox has been defined as a surprising situation that may occur when two populations are compared with respect to the incidence of some attribute: if the populations are separated in parallel into a set of descriptive categories, the population with higher overall incidence may yet exhibit a lower incidence within each such category \citep{Wagner82}.  With two categories Simpson's paradox occurs when $P(\textit{Y}=1|\textit{X}=1)>P(\textit{Y}=1|\textit{X}=0)$ while both $P(\textit{Y}=1|\textit{X}=1,\textit{W}=1)<P(\textit{Y}=1|\textit{X}=0,\textit{W}=1)$ and $P(\textit{Y}=1|\textit{X}=1,\textit{W}=0)<P(\textit{Y}=1|\textit{X}=0,\textit{W}=0)$.  We assume dichotomous variables and $P(\textit{Y}=1|\textit{X}=1)>P(\textit{Y}=1|\textit{X}=0)$ throughout this article. For this paper $\textit{Y}$ is the response or outcome, $\textit{X}$ records the treatment or exposure, and $\textit{W}$ is the covariate thought to be a confounding variable.

Given a $2\times 2\times 2$ contingency table and a uniform distribution on a simplex of cell probabilities that sum to unity, \citet{Pavlides09} show that approximately one in sixty contingency tables exhibit Simpson's paradox.  When Simpson's paradox occurs different conclusions are drawn depending on whether the aggregate or disaggregate data is interpreted.  \citet{Pearl14} argues that the paradox has been resolved through the use of causal diagrams.  When $\textit{W}$ is a confounding variable, with a causal impact on both $\textit{X}$ and $\textit{Y}$, then the disaggregate data should be used.

It is common for researchers to have data on $\textit{X}$ and $\textit{Y}$ while the confounding $\textit{W}$ remains unobserved.  Properties of $\textit{W}$ capturing how it relates to $\textit{X}$ and $\textit{Y}$ can be useful during sensitivity analysis even when these properties are not directly measured.  For example, \citet{Cornfield59} write ``The magnitude of the excess lung cancer risk ($\textit{Y}$) amongst cigarette smokers ($\textit{X}$) is so great that the results can not be interpreted as arising from an indirect association of cigarette smoking with some other agent or characteristic ($\textit{W}$), since this hypothetical agent would have to be at least as strongly associated with lung cancer as cigarette use; no such agent has been found or suggested.''  The argument relies in part on the ability of experts to assess possible magnitudes of association in the absence of data.

The logical condition, ``this hypothetical agent would have to be at least as strongly associated with lung cancer as cigarette use,'' is part of what is now referred to as Cornfield's condition.  \citet{Ding14,Ding16} have generalized Cornfield's condition, leading to the E-value for categorical sensitivity analysis.  The E-value is a quantity supplemental to the p-value and related to evidence for causality in observational studies \citep{Ding17}.  Alternatively, there is an odds ratio condition, derived here in this article from an optimal condition for continuous sensitivity analysis \citep[Figure 1, Remark A.1]{Knaeble17}. 

The purpose of this article is to introduce this odds ratio condition.  We conduct simulations to show that this odds ratio condition is as reliable as other commonly used conditions for basic categorical sensitivity analysis.  Relevant definitions are provided in Section 2, where our methodology is described.  The results of simulations are displayed in Section 3.  The utility of the odds ratio condition is demonstrated in Section 4, with an example case study highlighting how the odds ratio condition reflects the causal structure of confounding, not mediation.  A detailed discussion occurs in Section 6.  Supporting proofs are found in the appendix, where we derive the odds ratio condition and show it to be necessary for Simpson's paradox.
\section{Methods}
\label{meth}
Henceforth we refer to Simpson's paradox as defined in the introduction as Strong Simpson's Paradox.  Weak Simpson's Paradox occurs when either $P(\textit{Y}=1|\textit{X}=1,\textit{W}=1)<P(\textit{Y}=1|\textit{X}=0,\textit{W}=1)$ or $P(\textit{Y}=1|\textit{X}=1,\textit{W}=0)<P(\textit{Y}=1|\textit{X}=0,\textit{W}=0)$ but not both occur.  When the adjusted risk difference,
\begin{align*}
&P(\textit{W}=1)\left(P(\textit{Y}=1|\textit{X}=1,\textit{W}=1)-P(\textit{Y}=1|\textit{X}=0,\textit{W}=1)\right)+\\
&P(\textit{W}=0)\left(P(\textit{Y}=1|\textit{X}=1,\textit{W}=0)-P(\textit{Y}=1|\textit{X}=0,\textit{W}=0)\right),
\end{align*} 
is less than zero we say that an Adjusted Risk Difference Reversal has occurred.  An Adjusted Risk Difference Reversal occurs if and only if the adjusted relative risk (see $RR^{\text{true}}$ from \citet{Ding16}) is less than one.  Logically, Strong Simpson's Paradox implies Adjusted Risk Difference Reversal, and Adjusted Risk Difference Reversal implies Weak Simpson's Paradox. 

To detect reversal phenomena we have conditions we have various conditions specified as follows with $RR$ denoting the relative risk, $RD$ denoting the risk difference, and $OR$ denoting the odds ratio, and for each we use subscripts to specify the relevant variables, with the first subscript explanatory.  The Cornfield Condition is $(RR_{\textit{X}\textit{W}} \land RR_{\textit{W}\textit{Y}})> RR_{\textit{X}\textit{Y}}$.  The Risk Ratio Condition is $RR_{\textit{X}\textit{W}}RR_{\textit{W}\textit{Y}}/(RR_{\textit{X}\textit{W}}+RR_{\textit{W}\textit{Y}}-1)>RR_{\textit{X}\textit{Y}}$ \citep[Result 1]{Ding16}, and the Risk Difference Condition is $(\min\{RD_{\textit{X}\textit{W}},RD_{\textit{W}\textit{Y}}\}>RD_{\textit{X}\textit{Y}}) \land (\max\{RD_{\textit{X}\textit{W}},RD_{\textit{W}\textit{Y}}\}>\sqrt{RD_{\textit{X}\textit{Y}}})$ \citep[Theorem 1]{Ding14}. The Pearson Correlation Condition is $r(\textit{X},\textit{W})r(\textit{W},\textit{Y})>r(\textit{X},\textit{Y})$ \citep[Equation 3.24]{Cohen03}.  We refer to \begin{equation*}\label{Mxd} \left \{(\sqrt{OR_{\textit{X}\textit{W}} RR_{\textit{W}\textit{Y}}}+1)/(\sqrt{OR_{\textit{X}\textit{W}}} + \sqrt{RR_{\textit{W}\textit{Y}}}) \right \} ^2    > RR_{\textit{X}\textit{Y}}\end{equation*} \citep{Bross66, Bross67, Lee11} as the Mixed Condition because it utilizes an odds ratio and a relative risk factor.  The Odds Ratio Condition is given by \begin{align*}\label{ORC}&\left \{(\sqrt{OR_{\textit{W}\textit{X}}}-1)/(\sqrt{OR_{\textit{W}\textit{X}}}+1) \right \} \left \{ (\sqrt{OR_{\textit{W}\textit{Y}}}-1)/(\sqrt{OR_{\textit{W}\textit{Y}}}+1) \right \}\\&>r(\textit{X},\textit{Y}).\end{align*}  

The performance of a condition is measured by estimating the probability for a reversal or a non reversal given that the condition has occurred or not occurred.  These estimates are proportions computed from simulated $2\times 2\times 2$ contingency tables.  Following \citet{Pavlides09} each contingency table is randomly selected using a uniform distribution on the $7$-simplex 
$S=\{(p_1,...,p_8)\in\mathbb{R}^8:p_1+...+p_8=1\}$, where $p_1,...,p_8$ are cell probabilities.  We used the R command of \texttt{runif.rcomp} from the R package \texttt{compositions} \citep{Boog} to uniformly and randomly select points in $S$.  Given a point $(p_1,...,p_8)\in S$ we calculate $m=\min\{p_1,...,p_8\}$ and for $i=1,...,8$ the $i$th cell count of a random contingency table is determined as the smallest integer greater than $p_i/m$.  Each estimate $\hat{p}$ is obtained from a sample of $n>30,000$ contingency tables, and the formula $\sqrt{\hat{p}(1-\hat{p})/n}$ is used to compute standard errors.     
\section{Results}
Tables 1 and 2 show that the absence of stronger reversals can be reliably inferred from the absence of any condition, and the presence of weaker reversals can be reliably inferred from the presence of some of the conditions.  None of the conditions is sufficient for predicting stronger reversals, and the absence of a single condition is insufficient to rule out the weakest form of reversal. 

The Pearson Correlation Condition is the best indicator of an Adjusted Risk Difference Reversal.  However, with W unmeasured and categorical it may be difficult to estimate $r(\textit{W},\textit{X})$ and $r(\textit{W},\textit{Y})$.  The absence of the Odds Ratio Condition is the best indicator of the absence of an Adjusted Risk Difference Reversal.  Note that during application of the Odds Ratio Condition that $r(\textit{X},\textit{Y})$ can be computed unambiguously from the data.

When the Odds Ratio Condition fails to hold there is only a $2\%$ chance of an Adjusted Risk Difference Reversal.  The adjusted risk difference is an unbiased estimate for the average causal effect, if we assume that $\textit{W}$ suffices to completely control for confounding \citep{Ding16}.  Residual confounding in practice may be more concerning than uncertainty about a reversal.  The tables below may be used to select a condition with an associated probability sufficiently high so reversal uncertainty is tolerable in relation to the level of uncertainty inherent from residual confounding.
\begin{table}[h]
\centering
\caption{$P(\text{reversal}|\text{condition})$, assuming $OR_{xy}>1$.}
\label{T1}
\begin{tabular}{cccccc}
\hline
&&Strong Simpson&Adjusted Risk Difference&Weak Simpson\\
\hline
Cornfield&	& 			$.0805\pm .0012$	 		&$.3538\pm .0021$	&$.8607\pm .0016$\\
Risk Ratio&& 			$.0750\pm .0011$	 		&$.3794\pm .0019$	&$.8232\pm .0015$\\
Risk Difference&& 		$.1161\pm .0017$			&$.4959\pm .0027$	&$.9300\pm .0014$\\
Pearson Correlation&&	$.2004\pm .0020$			&$.7526\pm .0022$	&$.9983\pm .0002$\\
Mixed&& 				$.1626\pm .0017$			&$.6384\pm .0022$	&$.9483\pm .0010$\\
Odds Ratio&& 			$.1503\pm .0016$			&$.6132\pm .0021$	&$.9428\pm .0010$\\
\hline
\end{tabular}
\end{table}
\begin{table}[h]
\centering
\caption{$P(\text{not reversal}|\text{not condition})$, assuming $OR_{xy}>1$.}
\label{T2}
\begin{tabular}{cccccc}
\hline
&&Strong Simpson&Adjusted Risk Difference&Weak Simpson\\
\hline
Cornfield&&   			$.9912\pm .0001$ 		&$.9490\pm .0003$		&$.5364\pm .0007$\\
Risk Ratio&& 			$.9925\pm .0001$		&$.9615\pm .0003$		&$.5427\pm .0008$\\
Risk Difference&&		$.9915\pm .0001$		&$.9496\pm .0003$		&$.5285\pm .0007$\\
Pearson Correlation&&	$1\pm 0$				&$.9767\pm .0002$		&$.5395\pm .0007$\\
Mixed &&				$1\pm 0$				&$.9792\pm .0002$		&$.5451\pm .0007$\\
Odds Ratio&&			$1\pm 0$				&$.9818\pm .0002$		&$.5489\pm .0007$\\
\hline
\end{tabular}
\end{table}
\newpage
\section{Application}
\label{app}
The United States Centers for Disease Control and Prevention (CDC) has published ``Outcomes of Pregnancies with Laboratory Evidence of Possible Zika Virus Infection in the United States and the US Territories'' \citep{CDCa}.  They report that $91$ out of $1,784$ pregnancies with evidence of Zika virus infection have resulted in a live born infant with a birth defect.  This observed risk of $91/1784\approx.051$ is significantly higher than $1/33\approx .030$, which is the background risk for birth defects in the United States \citep{CDCb}.  With $\textit{X}$ indicating possible Zika virus infection on a population of pregnant mothers and $\textit{Y}$ indicating a liveborn infant with a birth defect, the relative risk can be expressed as $RR_{XY}=0.051/0.030=1.7$.  However, in \citet{CDCa} the CDC writes  ``we cannot determine whether individual defects were caused by Zika virus infection or other factors''.

With a large representative sample of controls (pregnant mothers without evidence of Zika virus infection, perhaps matched with cases on measurable covariate characteristics) we could have data in the form of a contingency table as shown in Table \ref{T3}, where the observed birth defect risk for the controls is $1/33$ as expected.  The numbers $501$ and $16533$ were chosen for demonstrational purposes.  The positive association observed in Table \ref{T3} may be spurious, due to a confounding factor $\textit{W}$ thought to cause both $\textit{X}$ and $\textit{Y}$.  For instance, $\textit{W}$ could indicate insufficient dietary intake of folic acid \citep{Weinhold09}.  According to \citet{Malone16} individuals infected with Zika virus may be asymptomatic, making it difficult to estimate $RR_{\textit{X}\textit{W}}$ or $RD_{\textit{X}\textit{W}}$, and in this categorical context correlation is difficult to estimate.  With our case-control setup $OR_{\textit{W}\textit{Y}}$ is easier to estimate than $RR_{\textit{W}\textit{Y}}$, because $OR_{\textit{W}\textit{Y}}=OR_{\textit{Y}\textit{W}}$ and $OR_{\textit{Y}\textit{W}}$ can be estimated retrospectively.  $OR_{\textit{W}\textit{X}}$ and $OR_{WY}$ reflect the causal structure of confounding.  We therefore use the Odds Ratio Condition in the following sensitivity analysis. 

Recall the probabilities from Table \ref{T1}, obtained from a random uniform distribution on the $7$-simplex.  Not all such selected tables are relevant here.  We modify the simulations and now select a $2\times 2\times 2$ table using a discrete uniform distribution on the space of all tables that collapse to Table \ref{T3}.  The resulting probabilities are shown in Tables \ref{T4} and \ref{T5}.  We may reasonably bound $OR_{\textit{W}\textit{Y}}$ below $1.4$ \citep{Ionescu-Ittu09}.  We compute $r(\textit{X},\textit{Y})=0.0328$ from Table \ref{T3}.  If we substitute $OR_{\textit{W}\textit{Y}}=1.4$ and $r(\textit{X},\textit{Y})=0.0328$ into the Odds Ratio Condition we then see that $OR_{\textit{W}\textit{X}}>5$ is required for a reversal.  Such a large odds ratio may be unreasonable \citep{Katona08}.  Both Table \ref{T2} and Table \ref{T5} then suggest that a reversal due to adjustment for confounding by folic acid deficiency is unlikely. 
\begin{table}[h]
\centering
\caption{A contingency table showing an association between Zika virus infection and birth defects.}
\label{T3}
\begin{tabular}{ccc}
\hline
&Lack of evidence&Evidence for infection \\
Birth defect presence&501&91\\
Birth defect absence&16533&1784\\
\hline
\end{tabular}
\end{table}
\begin{table}[h]
\centering
\caption{$P(\text{reversal}|\text{condition})$, assuming $OR_{xy}>1$, $OR_{wy}>1$, and conditional on Table \ref{T3}.}
\label{T4}
\begin{tabular}{cccccc}
\hline
&&Strong Simpson&Adjusted Risk Difference&Weak Simpson\\
\hline
Cornfield&& 			$.0544\pm .0005$ 	&$.2064\pm .0009$ 		&$.9586\pm .0004$\\
Risk Ratio&&			$.0543\pm .0005$ 	&$.2096\pm .0009$		&$.9400\pm .0005$\\
Risk Difference&&		$.0506\pm .0005$ 	&$.1890\pm .0008$		&$.9415\pm .0005$\\
Pearson Correlation&& 	$.0801\pm .0005$ 	&$.2572\pm .0008$ 		&$1\pm 0$ \\
Mixed &&				$.0716\pm .0005$ 	&$.2380\pm .0008$ 		&$.9872\pm .0002$\\
Odds Ratio&&			$.0430\pm .0003$	&$.1659\pm .0005$		&$.9089\pm .0004$\\
\hline
\end{tabular}
\end{table}
\begin{table}[h]
\centering
\caption{$P(\text{not reversal}|\text{not condition})$, assuming $OR_{xy}>1$, $OR_{wy}>1$, and conditional on Table \ref{T3}.}
\label{T5}
\begin{tabular}{cccccc}
\hline
&&Strong Simpson &Adjusted Risk Difference&Weak Simpson\\
\hline
Cornfield&&			$.9860\pm .0001$	&$.9397\pm .0003$		&$.3382\pm .0005$\\
Risk Ratio&&			$.9860\pm .0001$	&$.9406\pm .0003$		&$.3335\pm .0005$\\
Risk Difference&&		$.9858\pm .0001$	&$.9377\pm .0003$		&$.3393\pm .0005$\\
Pearson Correlation&&	$1\pm 0$	&$.9741\pm .0002$		&$.3844\pm .0006$\\
Mixed &&				$1\pm 0$	&$.9764\pm .0002$		&$.3969\pm .0006$\\
Odds Ratio&&			$1\pm 0$	&$.9910\pm .0001$		&$.4769\pm .0007$\\
\hline
\end{tabular}
\end{table}

\section{Discussion}
\label{dis}

The Odds Ratio Condition is derived in the appendix from the principle of least-squares.  Multiple regression minimizes mean square error within the class of linear models, and regression with categorical data can be reasonable in some situations \citep[Chapter 3]{AP09}.  Using indicator variables we say that a Least-Squares Reversal occurs when $\mathrm{sign}(\beta_{\textit{X}|\textit{W}})\neq \mathrm{sign}(\beta_{\textit{X}})$.  See \citet{Knaeble17} or the appendix for further details.  Probabilities for Least-Squares Reversals are not shown in the tables.  Strong Simpson's Paradox is sufficient for a Least-Squares Reversal.  The Pearson Correlation Condition and the Odds Ratio Condition are both necessary for a Least-Squares Reversal.  A Least-Squares Reversal is not necessary nor sufficient for an Adjusted Risk Difference Reversal.

The Cornfield Condition is necessary for an Adjusted Risk Difference Reversal if we assume both $RR_{xy} > 1$ and $RR_{xw} > 1$.  More complicated Risk Ratio Conditions \citep[Result 1; with their preceding definitions]{Ding16} and the Risk Difference Condition \citep[Theorem 1; lines 5 and 6]{Ding14} are necessary for an Adjusted Risk Difference Reversal, but the improvement requires specification of trivariate quantities: $RR_{\textit{W}\textit{Y}}$ conditional on both levels of $\textit{X}$.  This is close to outright specification of $P(\textit{Y}=1)$ conditional on $\textit{X}$ and $\textit{W}$, from which with adequate specification of frequencies the adjusted risk difference can be computed directly.  Here our simulations compared only those conditions built from simpler bivariate quantities.  Our analysis has been limited to the simplest case of categorical confounding by a single dichotomous covariate.  Result 1 of \citet{Ding16} is more general, and foundational for the newly introduced $E$-value \citep{Ding17}.  Our simulations suggest the possibility of an improved foundation.  

Our analysis has taken a possible interaction effect into account.  In the absence of an interaction effect Weak Simpsons Paradox can not occur, Adjusted Risk Difference Reversal and Strong Simpson's Paradox coincide, and since Strong Simpson's Paradox implies a Least-Squares Reversal, we have that the Pearson Correlation Condition and the Odds Ratio Condition are both necessary for an Adjusted Risk Difference Reversal.  With an interaction effect we can have an Adjusted Risk Difference Reversal without a Least-Squares Reversal, but according to Tables \ref{T1} and \ref{T3} this is rare, and possibly less concerning than uncertainty about residual confounding.  When an interaction effect is suspected we recommend simulations conditional on an observed data set as described in Section \ref{app}.  

In conclusion, we have shown how odds ratios can be used during sensitivity analysis on $2 \times 2$ contingency tables. The Odds Ratio Condition has been shown to be as reliable as other conditions that are based on risk ratios, risk differences, an odds ratio and a risk ratio, and correlation coefficients.  When the unmeasured binary third variable W is a confounder affecting both $\textit{X}$ and $\textit{Y}$, as opposed to a mediating variable affecting $\textit{Y}$ but affected by $\textit{X}$, then the odds ratio condition should be considered: $r(\textit{X},\textit{Y})$ can be computed from the data, while $OR_{\textit{W}\textit{X}}$ and $OR_{\textit{W}\textit{Y}}$ reflect the causal structure of confounding.

\section*{Appendix: Proofs}
When $\textit{Y}$ is regressed onto $\textit{X}$ we have $\beta_{\textit{X}}$ as the least-squares slope coefficient for $\textit{X}$.  We have assumed $\beta_\textit{X}>0$.
When $\textit{Y}$ is regressed onto $\textit{X}$ and $\textit{W}$ we have $\beta_{\textit{X}|\textit{W}}$ as the least-squares slope coefficient for $\textit{X}$, $\beta_{\textit{W}|\textit{X}}$ as the least-squares slope coefficient for $\textit{W}$, and $\beta_0$ as the intercept.  We say that a Least-Squares Reversal has occurred when $\beta_{\textit{X}|\textit{W}}<0$.
\label{der}
\begin{proposition}
$\emph{Strong Simpson's Paradox} \implies \emph{Least-Squares Reversal}$
\end{proposition}
\begin{proof}
Define $P_{00}=P(\textit{Y}=1|\textit{X}=0,\textit{W}=0)$, $P_{10}=P(\textit{Y}=1|\textit{X}=1,\textit{W}=0)$, $P_{01}=P(\textit{Y}=1|\textit{X}=0,\textit{W}=1)$, and $P_{11}=P(\textit{Y}=1|\textit{X}=1,\textit{W}=1)$.  Strong Simpson's Paradox implies $P_{00}>P_{10}$ and $P_{01}>P_{11}$.  Consider $\beta_0$, $\beta_{\textit{X}|\textit{W}}$, and $\beta_{\textit{W}|\textit{X}}$ as variables.  Define $\hat{P}_{00}=\beta_0$, $\hat{P}_{10}=\beta_0+\beta_{\textit{X}|\textit{W}}$, $\hat{P}_{01}=\beta_0+\beta_{\textit{W}|\textit{X}}$, and $\hat{P}_{11}=\beta_0+\beta_{\textit{X}|\textit{W}}+\beta_{\textit{W}|\textit{X}}$.  Let $n$ stand for the total number of observations.  The sum of the squares is a function of $(\beta_0,\beta_{\textit{X}|\textit{W}},\beta_{\textit{W}|\textit{X}})$, and it can be expressed as 
\begin{align*}
SS&=nP(\textit{X}=0,\textit{W}=0)(P_{00}(1-\beta_0)^2+(1-P_{00})\beta_0^2)\\
&+nP(\textit{X}=1,\textit{W}=0)(P_{10}(1-\beta_0-\beta_{\textit{X}|\textit{W}})^2+(1-P_{10})(\beta_0+\beta_{\textit{X}|\textit{W}})^2)\\
&+nP(\textit{X}=0,\textit{W}=1)(P_{01}(1-\beta_0-\beta_{\textit{W}|\textit{X}})^2+(1-P_{01})(\beta_0+\beta_{\textit{W}|\textit{X}})^2)\\
&+nP({\textit{X}}=1,{\textit{W}}=1)(P_{11}(1-\beta_0-\beta_{\textit{W}|\textit{X}}-\beta_{\textit{X}|\textit{W}})^2+(1-P_{11})(\beta_0+\beta_{\textit{W}|\textit{X}}+\beta_{\textit{X}|\textit{W}})^2).
\end{align*}
Within $SS$ each term is of the form 
$k(P(1-\hat{\textit{P}})^2+(1-P)\hat{P}^2)=k(\hat{\textit{P}}-P)^2+k(P-P^2)$, where $k>0$.  (As described in Section \ref{meth} our cell counts are nonzero with probability one.)  Therefore the fit improves whenever some of the distances $\{|\hat{P}_{00}-P_{00}|,|\hat{P}_{10}-P_{10}|,|\hat{P}_{01}-P_{01}|,|\hat{P}_{11}-P_{11}|\}$ decrease, as long as the others remain unchanged.

Suppose $\beta_{\textit{X}|\textit{W}}\geq0$.  $\beta_0$ and $\beta_{\textit{W}|\textit{X}}$ can be independently raised or lowered so as to improve the fit until $\beta_0\in [P_{10},P_{00}]$ or $\beta_0+\beta_{\textit{X}|\textit{W}}\in [P_{10},P_{00}]$, and $\beta_0+\beta_{\textit{W}|\textit{X}}\in [P_{11},P_{01}]$ or $\beta_0+\beta_{\textit{X}|\textit{W}}+\beta_{\textit{W}|\textit{X}}\in [P_{11},P_{01}]$.  Case 1: if $\beta_0\in [P_{10},P_{00}]$ and $\beta_0+\beta_{\textit{W}|\textit{X}}\in [P_{11},P_{01}]$ then $\exists \epsilon>0$ such that $(\beta_0,\beta_{\textit{X}|\textit{W}},\beta_{\textit{W}|\textit{X}})\mapsto (\beta_0,\beta_{\textit{X}|\textit{W}}-\epsilon,\beta_{\textit{W}|\textit{X}})$ improves the fit.  Case 2: if $\beta_0\not\in [P_{10},P_{00}]$ (which implies $\beta_0+\beta_{\textit{W}|\textit{X}}\in [P_{10},P_{00}]$) and $\beta_0+\beta_{\textit{W}|\textit{X}}\in [P_{11},P_{01}]$ then $\exists \epsilon>0$ such that $(\beta_0,\beta_{\textit{X}|\textit{W}},\beta_{\textit{W}|\textit{X}})\mapsto (\beta_0+\epsilon,\beta_{\textit{X}|\textit{W}}-\epsilon,\beta_{\textit{W}|\textit{X}}-\epsilon)$ improves the fit.  Case 3: if $\beta_0\not\in [P_{10},P_{00}]$ (which implies $\beta_0+\beta_{\textit{W}|\textit{X}}\in [P_{10},P_{00}]$) and $\beta_0+\beta_{\textit{W}|\textit{X}}\not\in [P_{11},P_{01}]$ (which implies $\beta_0+\beta_{\textit{X}|\textit{W}}+\beta_{\textit{W}|\textit{X}}\in [P_{11},P_{01}]$) then $\exists \epsilon>0$ such that $(\beta_0,\beta_{\textit{X}|\textit{W}},\beta_{\textit{W}|\textit{X}})\mapsto (\beta_0+\epsilon,\beta_{\textit{X}|\textit{W}}-\epsilon,\beta_{\textit{W}|\textit{X}})$ improves the fit.  In all three cases an improved fit is possible.  Our assumption $\beta_{\textit{X}|\textit{W}}\geq 0$ must therefore be faulty.  
\end{proof}
\begin{proposition}
$\emph{Least-Squares Reversal} \iff \emph{Pearson Correlation Condition}$.
\end{proposition}
\begin{proof}
From \citet[Equation 3.24]{Cohen03} we have \begin{equation}\label{one}\frac{s_\textit{X}}{s_\textit{Y}}\beta_{\textit{X}|\textit{W}}=\frac{r(\textit{X},\textit{Y})-r(\textit{W},\textit{X})r(\textit{W},\textit{Y})}{1-r(\textit{W},\textit{X})^2},\end{equation}
where $s_\textit{X}$ and $s_\textit{Y}$ are standard deviations of $\textit{X}$ and $\textit{Y}$ respectively.  Since $\frac{s_\textit{X}}{s_\textit{Y}}>0$ and $(1-r(\textit{W},\textit{X})^2)>0$ we see from (\ref{one}) that 
\begin{equation}\label{two}\beta_{\textit{X}|\textit{W}}<0 \iff r(\textit{X},\textit{W})r(\textit{W},\textit{Y})>r(\textit{X},\textit{Y}).\end{equation}
The right hand side of (\ref{two}) is the Pearson Correlation Condition.
\end{proof} 
\begin{lemma}
\label{lem2}
Given a correlation $r>0$ and an associated odds ratio $OR>1$ we have \[r<(\sqrt{OR}-1)/(\sqrt{OR}+1).\]
\end{lemma}
\begin{proof}
Define $\mathbf{x}=[0^a,1^b,0^c,1^d]$ and $\mathbf{y}=[1^a,1^b,0^c,0^d]$ from a $2\times2$ contingency table with nonzero cell frequencies $a$, $b$, $c$, and $d$.  The exponential notation indicates repeated entries.  The odds ratio is $OR=(bc)/(ad)$.  We have $bc>ad$.  For any permutation $\sigma$ of vector entries $r(\mathbf{x},\mathbf{y})=r(\sigma\mathbf{x},\sigma\mathbf{y})$.  The correlation coefficient can be expressed as \[r(\mathbf{x},\mathbf{y})=\frac{n\sum_{i=1}^nx_iy_i-\sum_{i=1}^nx_i\sum_{i=1}^ny_i}{\sqrt{n\sum_{i=1}^nx_i^2-(\sum_{i=1}^nx_i)^2}\sqrt{n\sum_{i=1}^ny_i^2-(\sum_{i=1}^ny_i)^2)}}.\]
In terms of $a$, $b$, $c$, and $d$, we have $r(\mathbf{x},\mathbf{y})=$
\begin{align*}
&\frac{(a+b+c+d)b-(b+d)(a+b)}{\sqrt{(a+b+c+d)(b+d)-(b+d)^2}\sqrt{(a+b+c+d)(a+b)-(a+b)^2}}\\
&=\frac{bc-ad}{\sqrt{(a+c)(b+d)(a+b)(c+d)}}\\
&=\frac{(bc)/(ad)-1}{\sqrt{(1+c/a)(1+b/d)(1+b/a)(1+c/d)}}\\
&=\frac{OR-1}{\sqrt{(1+OR+c/a+b/d)(1+OR+b/a+c/d)}}\\
&=\frac{OR-1}{\sqrt{(1+OR)^2+(c/a+b/d+b/a+c/d)(1+OR)+(c/a+b/d)(b/a+c/d)}}\\
&\leq \frac{OR-1}{\sqrt{(1+OR)^2+4\sqrt{(bc)/(ad)}(1+OR)+4(bc)/(ad)}}\\
&= \frac{OR-1}{\sqrt{(1+OR)^2+4\sqrt{OR}(1+OR)+4OR}}\\
&= \frac{OR-1}{\sqrt{(OR+2\sqrt{OR}+1)^2}}=\frac{OR-1}{OR+2\sqrt{OR}+1}=\frac{(\sqrt{OR}+1)}{(\sqrt{OR}+1)}\frac{(\sqrt{OR}-1)}{(\sqrt{OR}+1)}=\frac{\sqrt{OR}-1}{\sqrt{OR}+1}.
\end{align*}
The inequality follows from repeated application of $u^2 + v^2 \geq 2uv$.
\end{proof}
\begin{proposition}
$\emph{Pearson Correlation Condition} \implies \emph{Odds Ratio Condition}$.
\end{proposition}
\begin{proof}
We may without loss of generality swap the labels on the categories for $\textit{W}$ so as to ensure a positive association between $\textit{W}$ and $\textit{Y}$.  If the Pearson Correlation Condition holds then $r(\textit{X},\textit{W})$ must be positive as well.  Observe then that the odds ratios $OR_{\textit{X}\textit{W}}$ and $OR_{\textit{WY}}$ are both greater than one.  Lemma 4.3 thus applies and we have 
\begin{align*}
&r(\textit{X},\textit{W})r(\textit{W},\textit{Y})>r(\textit{X},\textit{Y})\implies\\
&\left \{(\sqrt{OR_{\textit{X}\textit{W}}}-1)/(\sqrt{OR_{\textit{X}\textit{W}}}+1) \right \} \left \{ (\sqrt{OR_{\textit{W}\textit{Y}}}-1)/(\sqrt{OR_{\textit{W}\textit{Y}}}+1) \right \}>r(\textit{X},\textit{Y})
\end{align*}
\end{proof}
The preceding propositions together prove the following theorem.
\begin{theorem}[a necessary condition for Simpson's paradox]
\label{th1}
\[\emph{Strong Simpson's Paradox}\implies \emph{Odds Ratio Condition}.\]
\end{theorem}

\end{document}